\pdfoutput=1
\documentclass[11pt]{article}

\usepackage{graphicx}

\usepackage{epsfig,
amssymb,amsmath,amsthm,float
}
\newtheorem{theorem}{Theorem}

\newtheorem{lemma}[theorem]{Lemma}

\def\T{\mathcal{T}}
\def\E{\mathbb{E}}

\def\eps{\epsilon}

\newcommand{\prob}[1]{{\sf Pr}\left(#1\right)}

\newcommand{\var}[1]{{\sf Var}\left(#1\right)}
\newcommand{\cov}[1]{{\sf Cov}\left(#1\right)}

\def\final{1}  
\def\iflong{\iffalse}
\ifnum\final=1  
\newcommand{\knote}[1]{[{\tiny karthik: \bf #1}]\marginpar{*}}
\newcommand{\snote}[1]{[{\tiny Santosh: \bf #1}]\marginpar{*}}
\newcommand{\sidecomment}[1]{}
\else 
\newcommand{\knote}[1]{}
\newcommand{\snote}[1]{}
\newcommand{\sidecomment}[1]{}
\fi  

\begin{document}

\title{\Large Algorithms for Implicit Hitting Set Problems}
\author{Karthekeyan Chandrasekaran \thanks{Georgia Institute of Technology. Supported in part by NSF awards AF-0915903 and AF-0910584. Email: {\tt karthe@gatech.edu,vempala@cc.gatech.edu}. }\\
\and Richard Karp \thanks{University of California, Berkeley. Email: {\tt karp@icsi.berkeley.edu}}\\
\and Erick Moreno-Centeno \thanks{Texas A\&M University. Email: {\tt e.moreno@tamu.edu}}\\
\and Santosh Vempala \footnotemark[1]}
\date{}

\maketitle

\begin{abstract}
A hitting set for a collection of sets is a set that has a non-empty intersection with each set in the collection; the hitting set problem is to find a hitting set of minimum cardinality. Motivated by instances of the hitting set problem where the number of sets to be hit is large, we introduce the notion of \emph{implicit hitting set problems}. In an implicit hitting set problem the collection of sets to be hit is typically too large to list explicitly; instead, an oracle
is provided which, given a set $H$, either determines that $H$ is a hitting set or returns a set that $H$ does not hit. We show a number of examples of classic implicit hitting set problems, and give a generic algorithm for solving such problems optimally. The main contribution of this paper is to show that this framework is valuable in developing approximation algorithms. We illustrate this methodology by presenting a simple on-line algorithm for the minimum feedback vertex set problem on random graphs. In particular our algorithm gives a feedback vertex set of size $n-(1/p)\log{np}(1-o(1))$ with probability at least $3/4$ for the random graph $G_{n,p}$ (the smallest feedback vertex set is of size $n-(2/p)\log{np}(1+o(1))$). We also consider a planted model for the feedback vertex set in directed random graphs. Here we show that a hitting set for a polynomial-sized subset of cycles is a hitting set for the planted random graph and this allows us to exactly recover the planted feedback vertex set.
\end{abstract}

\section{Introduction}

In the classic Hitting Set problem, we are given a universe $U$ of elements and a collection $\T$ of subsets $S_1, \ldots, S_m$ of $U$; the objective is to find a subset $H\subseteq U$ of minimum cardinality so that every subset $S_i$ in $\T$ contains at least one element from $H$. The problem is NP-hard \cite{karp-np-complete}, approximable to within $\log_2 |U|$ using a greedy algorithm, and has been studied for many interesting special cases.

There are instances of the hitting set problem where the number of subsets $|\T|$ to hit is exponential in the size of the universe. Consequently, obtaining a hitting set with approximation factor $\log_2|U|$ using the greedy algorithm which examines all subsets is unreasonable for practical applications. Our motivation is the possibility of algorithms that run in time polynomial in the size of the universe. In this paper, we introduce a framework that could be useful in developing efficient approximation algorithms for instances of the hitting set problem with exponentially many subsets to hit.

We observe that in many combinatorial problems, $\T$ has a succinct
representation that allows efficient verification of whether a candidate set
hits every subset in $\T$. Formally, in an implicit hitting set problem, the
input is a universe $U$ and a polynomial-time {\it oracle} that, given a set
$H$, either determines that $H$ is a hitting set or returns a subset that is
not hit by $H$. Thus, the collection $\T$ of subsets to hit is not specified
explicitly. The objective is to find a small hitting set by making at most
polynomial($|U|$) queries to the oracle. In Section 1.1, we show several
well-known problems that can be formulated as implicit hitting set problems.


We present a generic algorithm to obtain the optimal solution of implicit hitting set problems in Section 2. As this algorithm solves optimally the NP-hard (classic) hitting set problem as a subroutine, its worst-case running time is exponential as a function of $|U|$. The main purpose of stating the generic algorithm is to develop an intuition towards using the oracle. It suggests a natural way to use the oracle: first (1) propose a candidate hitting set $H$, then (2) use the oracle to check if the candidate set hits all the subsets, and if not obtain a subset $S$ that has not been hit, and finally (3) refine $H$ based on $S$ and repeat until a hitting set is found. 

The generic algorithm for the implicit hitting set problem is in fact a generalization of online algorithms for hitting set problems. Here, the ground set is specified in advance as before and the subsets to be hit arrive online. On obtaining a subset, the algorithm has to decide which new element to include in the hitting set and commit to the element. Thus, the online algorithm is restricted in that the refinement procedure can only add elements.  Moreover, only those subsets that have not been hit by the candidate set are revealed online thereby saving the algorithm from having to examine all subsets in $\T$. This is similar to the mistake bound learning model \cite{littlestone88}.


We apply the implicit hitting set framework and specialize the generic
algorithm to the \emph{Minimum Feedback Vertex Set} (FVS) problem: given a
graph $G(V,E)$, find a subset $S\subseteq V$ of smallest cardinality so that
every cycle in the graph contains at least one vertex from $S$. Although the
number of cycles could be exponential in the size of the graph, one can
efficiently check whether a proposed set $H$ hits all cycles ({\it i.e.,} is a
feedback vertex set) or find a cycle that is not hit by $H$ using a
breadth-first search procedure after removing the subset of vertices $H$ from
the graph. The existence of a polynomial time oracle shows that it is an
instance of the implicit hitting set problem.

The main focus of this paper is to develop algorithms that find nearly optimal hitting sets in random graphs or graphs with planted feedback vertex sets, by examining only a polynomial number of cycles. For this to be possible, we need the oracle to pick cycles that have not yet been hit in a natural yet helpful manner. If the oracle is adversarial, this could force the algorithm to examine almost all cycles. We consider two natural oracles: one that picks cycles in breath-first search (BFS) order and another that picks cycles according to their size.

We prove that if cycles in the random graph $G_{n,p}$ are obtained in a breadth-first search ordering, there is an efficient algorithm that examines a polynomial collection $\T'$ of cycles to build a nearly optimal feedback vertex set for the graph. The algorithm builds a solution iteratively by (1) proposing a candidate for a feedback vertex set in each iteration, (2) finding the next cycle that is not hit in a breadth-first ordering of all cycles, (3) augmenting the proposed set and repeating. A similar result for directed random graphs using the same algorithm follows by ignoring the orientation of the edges. Our algorithm is an online algorithm i.e., it commits to only adding and not deleting vertices from the candidate feedback vertex set.

It is evident from our results that the size of the feedback vertex set in both directed and undirected random graphs is close to $n$, for sufficiently large $p$. This motivates us to ask if a smaller planted feedback vertex set in random graphs can be recovered by using the implicit hitting set framework. This question is similar in flavor to the well-studied planted clique problem \cite{jerrumClique92,sudakov98,frieze08}, but posed in the framework of implicit hitting set problems. We consider a natural planted model for the feedback vertex set problem in directed graphs. In this model, a subset of $\delta n$ vertices, for some constant $0<\delta\leq 1$, is chosen to be the feedback vertex set. The subgraph induced on the complement is a random directed acyclic graph (DAG) and all the other arcs are chosen with probability $p$ independently. The objective is to recover the planted feedback vertex set. We prove that the optimal hitting set for cycles of bounded size is the planted feedback vertex set. Consequently, ordering the cycles according to their sizes and finding an approximately optimal hitting set for the small cycles is sufficient to recover the planted feedback vertex set. This also leads to an online algorithm when cycles are revealed in increasing order of their size with ties broken arbitrarily.

We conclude this section with some well-known examples of implicit hitting set problems.

\subsection{Implicit Hitting Set Problems}

An {\it implicit hitting set problem} is one in which, for each instance, the
set of subsets is not listed explicitly but instead is specified implicitly by
an {\it oracle}: a polynomial-time algorithm which, given a set $H \subset U$,
either certifies that $H$ is a hitting set or returns a subset that is not hit
by $H$.

\noindent Each of the following is an implicit hitting set problem:
\begin{itemize}
\item {\bf Feedback Vertex Set in a Graph or Digraph}\\
Ground Set: Set of vertices of graph or digraph $G$.\\
Subsets: Vertex sets of simple cycles in $G$.

\item {\bf Feedback Edge Set in a Digraph}\\
Ground Set: Set of edges of digraph $G$.\\
Subsets: Edge sets of simple cycles in $G$.

\item {\bf Max Cut}\\
Ground Set: Set of edges of graph $G$.\\
Subsets: Edge sets of simple odd cycles in $G$.

\item {\bf k-Matroid Intersection}\\
Ground Set: Common ground set of $k$ matroids.\\
Subsets: Subsets in the $k$ matroids.

\item {\bf Maximum Feasible Set of Linear Inequalities}\\
Ground Set: A finite set of linear inequalities.\\
Subsets: Minimal infeasible subsets of the set of linear inequalities.

\item {\bf Maximum Feasible Set of Equations of the Form $x_i - x_j = c_{ij}\ (\!\!\!\mod q)$}\\
This example is motivated by the Unique Games Conjecture.

\item {\bf Synchronization in an Acyclic Digraph}\\
Ground Set: A collection $U$ of pairs of vertices drawn from the vertex set of
an acyclic digraph $G$.\\
Subsets: Minimal collection $C$ of pairs from $U$ with the property that, if
each pair in $C$ is contracted to a single vertex, then the resulting digraph
contains a cycle.
\end{itemize}


\noindent {\bf Organization.} In Section 2, we present a generic algorithm for the optimal solution of implicit hitting set problems. Then, we focus on specializing this algorithm to obtain small feedback vertex sets in directed and undirected random graphs. We analyze the performance of this algorithm in Section 3. We then consider a planted model for the feedback vertex set problem in directed random graphs. In Section 4, we give an algorithm to recover the planted feedback vertex set by finding an approximate hitting set for a polynomial-sized subset of cycles. We prove a lower bound for the size of the feedback vertex set in random graphs in Section 5. We state our results more precisely in the next section.

\subsection{Results for Feedback Vertex Set Problems}
We consider the feedback vertex set problem for the random graph $G_{n,p}$, a graph on $n$ vertices in which each edge is chosen independently with probability $p$. Our main result here is that a simple augmenting approach based on ordering cycles according to a breadth-first search (Algorithm Augment-BFS described in the next section) has a strong performance guarantee.

\begin{theorem}\label{theorem:undirected-MFVS}
For $G_{n,p}$, such that $p=o(1)$, there exists a polynomial time algorithm that produces a feedback vertex set of size at most $n-(1/p)\log{(np)}(1-o(1))$ with probability at least $3/4$.
\end{theorem}
Throughout, $o(1)$ is with respect to $n$. We complement our upper bound with a lower bound on the feedback vertex set for $G_{n,p}$ obtained using simple union bound arguments.
\begin{theorem}\label{theorem:undirected-FVSlowerbound}
Let $r=\frac{2}{p}\log{(np)}(1+o(1))+1$. If $p<1/2$, then every subgraph induced by any subset of $r$ vertices in $G_{n,p}$ contains a cycle with high probability.
\end{theorem}

This gives an upper bound of $r-1$ on the maximum induced acyclic subgraph of $G_{n,p}$. So, the size of the minimum feedback vertex set for $G_{n,p}$ is at least $n-r+1=n-(2/p)\log{np}$. A result of Fernandez de la Vega \cite{delavega96} shows that $G_{n,p}$ has an induced tree of size at least $(2/p)\log{np}(1-o(1))$, when $p=o(1)$. This gives the best possible existential result: there exists a feedback vertex set of size at most $n-(2/p)\log{np}(1-o(1))$ with high probability in $G_{n,p}$, when $p=o(1)$. We note that this result is not algorithmic; Fernandez de la Vega gives a greedy algorithm to obtain the largest induced tree of size $(1/p)\log{np}(1-o(1))$ in \cite{delavega86}. This algorithm is based on growing the induced forest from the highest labeled vertex and does not fall in the implicit hitting set framework (when the graph is revealed as a set of cycles). In contrast, our main contribution to the FVS problem in random graphs is showing that a simple breadth-first ordering of the cycles is sufficient to find a nearly optimal feedback vertex set. We also note that our algorithm is an online algorithm with good performance guarantee when the cycles are revealed according to a breadth-first ordering. Improving on the size of the FVS returned by our algorithm appears to require making progress on the long-standing open problem of finding an independent set of size $((1+\eps)/p)\log {np}$ in $G_{n,p}$. Assuming an optimal algorithm for this problem leads to an asymptotically optimal guarantee matching Fernandez de la Vega's existential bound.

Next, we turn our attention to the directed random graph $D_{n,p}$ on $n$ vertices. The directed random graph $D_{n,p}$ is obtained as follows: choose a set of undirected edges joining distinct elements of $V$ independently with probability $2p$. For each chosen undirected edge $\{u,v\}$, orient it in one of the two directions $\{u\rightarrow v, v\rightarrow u\}$ in $D_{n,p}$ with equal probability.

The undirected graph $G_D$ obtained by ignoring the orientation of the edges in $D_{n,p}$ is the random graph $G(n,2p)$. Moreover, a feedback vertex set in $G_D$ is also a feedback vertex set for $D_{n,p}$. Therefore, by ignoring the orientation of the arcs, the Augment-BFS algorithm as applied to undirected graphs can be used to obtain a feedback vertex set of size at most $n-(1/2p)\log{(2np)}$ with probability at least $3/4$. A theorem of Spencer and Subramanian \cite{spencer-subramanian} gives a nearly matching lower bound on the size of the feedback vertex set in $D_{n,p}$.
\begin{theorem}\cite{spencer-subramanian}\label{theorem:directed-FVSlowerbound}
Consider the random graph $D_{n,p}$, where $np\geq W$, for some fixed constant $W$. Let $r=(2/\log{(1-p)^{-1}})(\log{(np)}+3e)$. Every subgraph induced by any subset of $r$ vertices in $G$ contains a cycle with high probability.
\end{theorem}

It is evident from the results above that the feedback vertex set in a random graph contains most of its vertices when $p=o(1)$. This motivates us to ask if a significantly smaller ``planted'' feedback vertex set in a random graph can be recovered with the implicit hitting set framework. In order to address this question, we present the following planted model.

The planted directed random graph ${D}_{n,\delta,p}$ on $n$ vertices for $0<\delta\leq 1$ is obtained as follows: Choose $\delta n$ vertices arbitrarily to be the planted subset $P$. Each pair $(u,v)$ where $u\in P, v\in V$, is adjacent independently with probability $2p$ and the corresponding edge is oriented in one of the two directions $\{u\rightarrow v, v\rightarrow u\}$ in $D_{n,\delta,p}$ with equal probability. The arcs between vertices in $V\setminus P$ are obtained in the following manner to ensure that the subgraph induced on $V\setminus P$ is a DAG: Pick an arbitrary permutation of the vertices in $V\setminus P$. With the vertices ordered according to this permutation, each forward arc is present with probability $p$ independently; no backward arcs occur according to this ordering.

\begin{figure}[H]
\label{fig:planted-model}
\begin{center}
\psfig{file=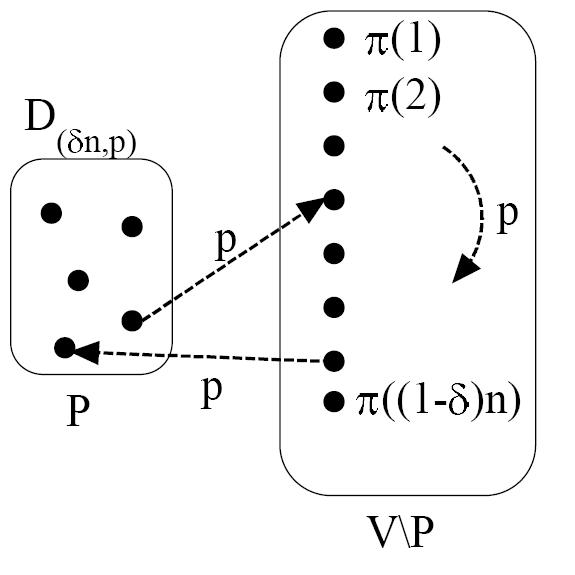,width=1.6in,height=1.6in}
\end{center}
\caption{Planted Model}
\end{figure}

We prove that for graphs $D_{n,\delta,p}$, for large enough $p$, it is sufficient to hit cycles of small size to recover the planted feedback vertex set. For example, if $p\geq C_0/n^{1/3}$ for some absolute constant $C_0$, then it is sufficient to find the best hitting set for triangles in $D_{n,\delta,p}$. This would be the planted feedback vertex set. We state the theorem for cycles of length $k$.
\begin{theorem}\label{theorem:planted-directedFVS}
Let $D$ be a planted directed random graph $D_{n,\delta,p}$ with planted feedback vertex set $P$, where $p\geq C/n^{1-2/k}$ for some constant $C$ and $0<\delta\leq 9/19$. Then, with high probability, the smallest hitting set for the set of cycles of size $k$ in $D$ is the planted feedback vertex set $P$.
\end{theorem}
Thus, in order to recover the planted feedback vertex set, it is sufficient to obtain cycles in increasing order of their sizes and find the best hitting set for the subset of all cycles of size $k$. Moreover, the expected number of cycles of length $k$ is at most $(nkp)^{k}=\text{poly}(n)$ for the mentioned range of $p$ and constant $k$. Thus, we have a polynomial-sized collection $\T'$ of cycles, such that the optimal hitting set for $\T'$ is also the optimal hitting set for all cycles in $D_{n,\delta,p}$.

However, finding the smallest hitting set is NP-hard even for triangles. We give an efficient algorithm to recover the planted feedback vertex set using an \emph{approximate} hitting set for the small cycles.
\begin{theorem}\label{theorem:algorithm-planted-directedFVS}
Let $D$ be a planted directed random graph $D_{n,\delta,p}$ with planted feedback vertex set $P$, where $p\geq C/n^{1-2/k}$ for some constant $C$ and $k\geq 3$, $0<\delta\leq 1/2k$. Then, there exists an algorithm to recover the planted feedback vertex set $P$ with high probability; this algorithm has an expected running time of $(nkp)^{O(k)}$.
\end{theorem}

\section{Algorithms}
In this section, we mention a generic algorithm for implicit hitting set problems. We then focus on specializing this algorithm to the feedback vertex set problems in directed and undirected graphs.

\subsection{A Generic Algorithm}
We mention a generic algorithm for solving instances of the implicit hitting
set problem optimally with the aid of an oracle and a subroutine for the exact solution of (explicit) hitting set problems. The guiding principle is to build up a short list of important subsets that dictate the solution, while limiting the number of times the subroutine is invoked, since its computational cost is
high.

A set $H \subset U$ is called {\it feasible} if it is a hitting set for the
implicit hitting set problem, and {\it optimal} if it is feasible and of
minimum cardinality among all feasible  hitting sets. Whenever the  oracle reveals
that a set $H$ is not feasible, it returns  $c(H)$, a subset that H does not
hit. Each generated subset $c(H)$ is added to a growing list
$\Gamma$ of subsets. A set $H$ is called $\Gamma$-feasible if it hits every
subset in $\Gamma$ and $\Gamma$-optimal if it is $\Gamma$-feasible and of
minimum cardinality among all $\Gamma$-feasible subsets. If a $\Gamma$-optimal set $K$
is feasible then it is necessarily optimal since $K$ is a valid hitting set for the implicit hitting set problem which contains subsets in $\Gamma$, and $K$ is the minimum hitting set for subsets in $\Gamma$. Thus the goal of the algorithm is to construct a feasible $\Gamma$-optimal set. 

\noindent {\bf Generic Algorithm}\\
\noindent Initialize $\Gamma \leftarrow \emptyset$.
\begin{enumerate}
\item Repeat:
\begin{enumerate}
\item $H\leftarrow U$.
\item Repeat while there exists a $\Gamma$-feasible set $H' = (H \cup X) - Y$ such that $X,Y\subseteq U$, $|X|<|Y|$:
	\begin{enumerate}
	\item If $H'$ is feasible then $H\leftarrow H'$; else $\Gamma \leftarrow \Gamma \cup \{c(H')\}$.
	\end{enumerate}
\item Construct a $\Gamma$-optimal set $K$.
\item If $|H|=|K|$ then return $H$ and halt ($H$ is optimal); if $K$ is feasible then return $K$ and halt ($K$ is optimal); else $\Gamma \leftarrow \Gamma \cup \{c(K)\}$.
\end{enumerate}
\end{enumerate}

\noindent {\bf Remark 1}. Since the generic algorithm solves optimally an NP-hard problem as a subroutine, its worst-case execution time is exponential in $|U|$. Its effectiveness in practice depends on the choice of the missed subset that the oracle returns.

A companion paper \cite{km} describes successful computational experience with
an algorithm that formulates a multi-genome alignment problem as an implicit
hitting set problem, and solves it using a specially tailored variant of the 
generic algorithm.

\begin{center}
\fbox{\parbox{6.7in}{
\begin{minipage}{6.5in}
\begin{tt}

\noindent {\bf Algorithm Augment-BFS}
\begin{enumerate}
\item Start from an arbitrary vertex as a surviving vertex. Initialize i=1.
\item Repeat:
\begin{enumerate}
\item Obtain cycles induced by one step BFS-exploration of the surviving vertices at depth i. Delete vertices at depth i+1 that are present in these cycles. Declare the remaining vertices at depth i+1 as surviving vertices.
\item If no vertices at depth i+1 are surviving vertices, terminate and output the set of all deleted vertices.
\item i=i+1.
\end{enumerate}
\end{enumerate}

\end{tt}
\end{minipage}
}}
\end{center}

\subsection{Algorithm Augment-BFS}
In this section, we give an algorithm to find the feedback vertex set in both undirected and directed graphs. Here, we use an oracle that returns cycles according to a breadth-first search ordering. Instead of the exact algorithm for the (explicit) hitting set problem, as suggested in the generic algorithm, we use a simpler strategy of picking a vertex from each missed cycle. Essentially, the algorithm considers cycles according to a breadth-first search ordering and maintains an induced tree on a set of vertices denoted as surviving vertices. The vertices deleted in the process will constitute a feedback vertex set. Having built an induced tree on surviving vertices up to a certain depth $i$, the algorithm is presented with cycles obtained by a one-step BFS exploration of the surviving vertices at depth $i$. For each such cycle, the algorithm picks a vertex at depth $i+1$ to delete. The vertices at depth $i+1$ that are not deleted are added to the set of surviving vertices, thereby leading to an induced tree on surviving vertices up to depth $i+1$.

\noindent {\bf Remark 2}. Although a very similar algorithm can be used for other variants of the feedback set problem, we note that these problems in random graphs turn out to be easy. For example, the feedback edge set problem is equivalent to the maximum spanning tree problem, while the feedback arc set problem has tight bounds for random graphs using very simple algorithms.

\section{Feedback Vertex Set in Random Graphs}
In this section, we show that Augment-BFS can be used to find a nearly optimal feedback vertex set in the undirected random graph $G_{n,p}$. Our main contribution is a rigorous analysis of the heuristic of simple cycle elimination in BFS order. We say that a vertex $v$ is a {\bf \emph{unique}} neighbor of a subset of vertices $L$ if and only if $v$ is adjacent to exactly one vertex in $L$.

In Algorithm Augment-BFS, we obtain induced cycles in BFS order having deleted the vertices from the current candidate FVS $S$. We refine the candidate
FVS $S$ precisely as follows to obtain an induced BFS tree with unit increase in height: Consider the set $c(S)$ of cycles obtained by one-step BFS exploration from the set of vertices at current depth. Let $K$ denote the set of unexplored vertices in the cycles in $c(S)$ ($K$ is a subset of the vertices obtained by one-step BFS exploration from the set of vertices at current depth). Among the vertices in $K$ include all non-unique neighbors of the set of vertices at current depth into $S$. Find a large independent set in the subgraph induced by the unique neighbors $R\subseteq K$ of the set of vertices at current depth. Include all vertices in $R$ that are not in the independent set into $S$. This iterative refinement process is a natural adaptation of the idea behind the generic algorithm to the feedback vertex set problem where one collects a subset of cycles to find a hitting set $H$ for these cycles and proposes $H$ as the candidate set to obtain more cycles that have not been hit.

\begin{figure}[H]
\label{fig:bfs-exploration}
\begin{center}
\psfig{file=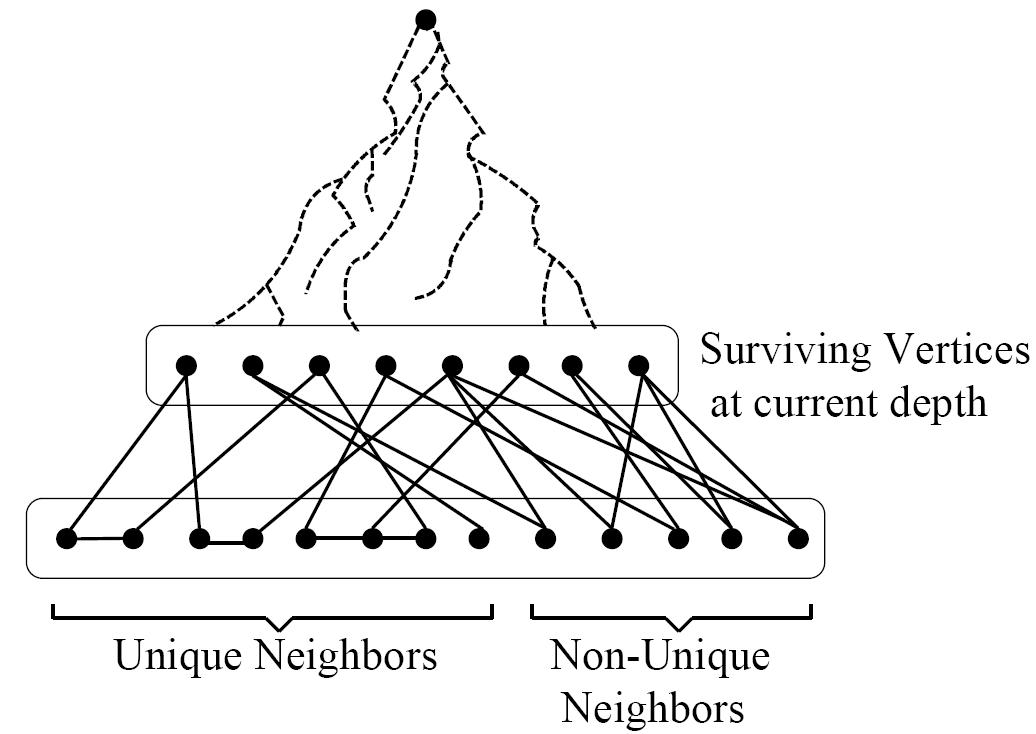,width=3in,height=2.2in}
\end{center}
\caption{BFS Exploration}
\end{figure}

Essentially, the algorithm maintains an induced BFS tree by deleting vertices to remove cycles. The set of deleted vertices form a FVS. Consequently at each level of the BFS exploration, one would prefer to add as many vertices from the next level $K$ as possible maintaining the acyclic property. One way to do this is as follows: Delete all the non-unique neighbors of the current level from $K$ thus hitting all cycles across the current and next level. There could still be cycles using an edge through the unique neighbors. To hit these, add a large independent set from the subgraph induced by the unique neighbors and delete the rest. Observe that this induced subgraph is a random graph on a smaller number of vertices. However, even for random graphs, it is open to find the largest independent set efficiently and only a factor $2$ approximation is known.

In our analysis, instead of using the two approximate algorithm for the independent set problem, we use the simple heuristic of deleting a vertex for each edge that is present in the subgraph to find an independent set at each level. In order to lower bound the size of the induced tree, it suffices to consider growing the BFS-tree up to a certain height $T$ using this heuristic and then using the $2$-approximate algorithm for independent set at height $T$ to terminate the algorithm. The size of the induced tree obtained using Algorithm Augment-BFS is at least as large as the one produced by the process just described. To simplify our analysis, it will be useful to restate the algorithm as Algorithm Grow-induced-BFS.

\begin{center}
\fbox{\parbox{6.6in}{
\begin{minipage}{6.4in}
\begin{tt}
\noindent {\bf Algorithm Grow-induced-BFS}
\begin{enumerate}
\item Start from an arbitrary vertex $v$ at level $0$, set $L_0=\{v\}$. Mark $v$ as exposed. Fix $c:=np$.
\item Explore levels $i=0,\cdots, T-1$, where $T=\left \lceil{\frac{\ln{(1/16p)}-\ln{\ln{(1/16p)}}}{\ln{(c+20\sqrt{c})}}}\right \rceil$ in BFS order as follows:
\begin{enumerate}
\item Let $K_{i+1}$ be the subset of neighbors of $L_{i}$ among the unexposed vertices, where $L_{i}$ is the set of surviving vertices at level $i$.
\item Mark the vertices in $K_{i+1}$ as exposed.
\item Let $R_{i+1}\subseteq K_{i+1}$ be the subset of vertices in $K_{i+1}$ that are {\bf \emph {unique}} neighbors of $L_i$.
\item For every edge $(u,v)$ that is present between vertices $u,v\in R_{i+1}$, add either $u$ or $v$ to $W_{i+1}$.
\item Set $L_{i+1}=R_{i+1}\setminus W_{i+1}$.\\
 \emph{(The set of surviving vertices at level $i+1$, namely $L_{i+1}$ is an independent set in the subgraph induced by $R_{i+1}$.)}
\end{enumerate}
\item On obtaining $L_{T-1}$, set $R_{T}$ = {\bf \emph{unique}} neighbors of $L_T$ among the unexposed vertices. In the subgraph induced by $R_{T}$, find an independent set $L_{T}$ as follows.
\begin{enumerate}
\item Fix an arbitrary ordering of the vertices of $R_{T}$. Repeat while $R_{T}\neq \emptyset$:
\begin{itemize}
\item Add the next vertex $v\in R_{T}$ to $L_{T}$. Let $N(v)$= neighbors of $v$ in $R_{T}$. Set $R_{T}\leftarrow R_{T}\setminus N(v)$.
\end{itemize}
\end{enumerate}
\item Return $S=V\setminus \cup_{i=0}^{T} L_{i}$ as the feedback vertex set.
\end{enumerate}
\end{tt}
\end{minipage}
}}
\end{center}
\noindent 

We remark that improving the approximation factor of the largest independent set problem in $G_{n,p}$ would also improve the size of the FVS produced. Our analysis shows that most of the vertices in the induced BFS tree get added at depth $T$ as an independent set.
Moreover, the size of this independent set is close to $(2/p)\log{np}(1-o(1))$. Consequently, any improvement on the approximation factor of the largest independent set problem in $G_{n,p}$ would also lead to improving the size of the independent set found at depth $T$. This would increase the number of vertices in the induced BFS tree and thereby reduce the number of vertices in the feedback vertex set.

Observe that Algorithm Grow-induced-BFS can be used for the directed random graph $D_{n,p}$ by ignoring the orientation of the edges to obtain a nearly optimal feedback vertex set. Such a graph obtained by ignoring the orientation of the edges is the random graph $G(n,2p)$. Further, a FVS in such a graph is also a FVS in the directed graph. Consequently, we have the following theorem.
\begin{theorem}\label{directed-FVS}
For $D_{n,p}$, there exists a polynomial time algorithm that produces a FVS of size at most $n-(1/2p)(\log{(np)}-o(1))$ with probability at least $3/4$.
\end{theorem}
By Theorem \ref{theorem:directed-FVSlowerbound}, we see that the algorithm is nearly optimal for directed random graphs.

Next, we analyze Algorithm Grow-induced-BFS to find the size of the FVS that it returns. For $i=0,\cdots,T$, let
$L_i$ be the set of surviving vertices at level $i$ with $l_i:=|L_i|$,
$R_{i+1}$ be the set of \emph{unique} neighbors of $L_i$ with $r_{i+1}:=|R_{i+1}|$, and
$U_{i}$ be the set of unexposed vertices of the graph after $i$ levels of BFS exploration with $u_{i}:=|U_{i}|$. Observe that $U_i:=V\setminus (L_0\cup_{j=1}^{i} K_i)$.

We will need the following theorem due to Frieze \cite{frieze-ind-set}, about the size of the independent set.
\begin{theorem}\cite{frieze-ind-set} \label{theorem:large-ind-set}
Let $d=np$ and $\eps>0$ be fixed. Suppose $d_{\eps}\leq d=o(n)$ for some sufficiently large fixed constant $d_{\eps}$. Then, almost surely, the size of the independent set in $G_{n,p}$ is at least
\[
\left(\frac{2}{p}\right)(\log{np}-\log{\log{np}}-\log{2}+1-0.5\eps).
\]
\end{theorem}

\subsection{Large Set of Unique Neighbors}
The following lemma gives a concentration of the number of surviving vertices, unexposed vertices and unique neighbors to survivors at a particular level. It shows that upon exploring $t$ levels according to the algorithm, the number of surviving vertices at the $t$-th level, $l_t$, is not too small while the number of unexposed vertices, $u_t$, is large. It also shows a lower bound on the number of unique neighbors $r_{t+1}$ to a level of survivors. This fact will be used in proving Theorem \ref{theorem:undirected-MFVS}.
\begin{lemma}\label{lemma:largesurvivors}
Let $c:=np$ and $T$ be the largest integer that satisfies $16Tp(c+20\sqrt{c})^{T-1}\leq 1/2$.
Then, with probability at least $3/4$, $\forall t \in \{0,1,\cdots,T-1\}$,
\begin{enumerate}
\item
\begin{align*}
u_t &\leq \left(n-\frac{1}{4}\sum_{i=0}^t(c-20\sqrt{c})^i \right)\left(1+\sqrt{\frac{\ln{\ln{n}}}{n}}\right) \\
u_t &\geq (n-\sum_{i=0}^t(c+20\sqrt{c})^i)\left(1-\sqrt{\frac{\ln{\ln{n}}}{n}}\right)
\end{align*}
\item
\begin{align*}
l_t &\leq \left(c+20\sqrt{c} \right)^t\\
l_t &\geq \left(c-20\sqrt{c}\right)^t(1-16Tp(c+20\sqrt{c})^t)\\
	& \ \ \ \times \left(1-\frac{\sum_{i=0}^t(c+20\sqrt{c})^i}{n}\right)
\end{align*}
\item
\begin{align*}
r_t &\leq (c+20\sqrt{c})^{t+1}\left(1+\sqrt{\frac{\ln{\ln{n}}}{n}}\right)\\
r_t &\geq \frac{(c-20\sqrt{c})^{t+1}}{4} \left(1-\frac{\sum_{i=0}^{t+1}(c+20\sqrt{c})^i}{n}\right)\\
  & \ \ \ \times \left(1-\sqrt{\frac{\ln{\ln{n}}}{n}}\right)
\end{align*}
\end{enumerate}
\end{lemma}

Now, we are ready to prove Theorem \ref{theorem:undirected-MFVS}.

\subsection{Proof of main theorem}
\begin{proof} [Proof of Theorem \ref{theorem:undirected-MFVS}]
Our objective is to use the fact that the size of the surviving set of vertices is large when the algorithm has explored $T-1$ levels. Moreover, the number of unexposed vertices is also large. Thus, there is a large independent set among the \emph{unique} neighbors of the surviving vertices. This set along with the surviving vertices up to level $T-1$ will form a large induced tree. We will now prove that the size of the independent set among the \emph{unique} neighbors of $L_{T-1}$ is large.

By Theorem \ref{theorem:large-ind-set}, if $r_{T}p>d_{\eps}$ for some constant $d_{\eps}$ and $r_{T}p=o(r_{T})$, then there exists an independent set of size $(2/p)\log{(r_{T}p)}(1-o(1))$. It suffices to prove that $r_{T}$ is large and is such that $r_Tp>d_{\eps}$.

Note that the choice of $T=\left \lceil \frac{\ln{(1/16p)}-\ln{\ln{(1/16p)}}}{\ln{(c+20\sqrt{c})}}\right \rceil$ used in the algorithm satisfies the hypothesis of Lemma \ref{lemma:largesurvivors}. Therefore, using Lemma \ref{lemma:largesurvivors}, with probability at least $3/4$, we have
\begin{align*}
r_{T}&\geq \frac{(c-20\sqrt{c})^{T}}{4} \left(1-\frac{\sum_{i=0}^{T}(c+20\sqrt{c})^i}{n}\right)\\
 & \ \ \ \times \left(1-\sqrt{\frac{\ln{\ln{n}}}{n}}\right)\\
&\geq \frac{(c-20\sqrt{c})}{64p}\left(1-\frac{\sum_{i=0}^{T}(c+20\sqrt{c})^i}{n}\right)\\
 & \ \ \ \times \left(1-\sqrt{\frac{\ln{\ln{n}}}{n}}\right)\\
&\geq \frac{c-20\sqrt{c}}{2^8p} \geq \frac{d_{\eps}}{p}
\end{align*}
for sufficiently large $c$ since
\[
\left(1-\frac{\sum_{i=0}^{T}(c+20\sqrt{c})^i}{n}\right)\left(1-\sqrt{\frac{\ln{\ln{n}}}{n}}\right)\geq \frac{15}{16}\cdot\frac{1}{2}\\.
\]

Consequently, by Theorem \ref{theorem:large-ind-set}, there exists an independent set of size at least $\left({2}/{p}\right)\log {(r_{T}p)}(1-o(1))$. Moreover, step 3 of the algorithm finds a $2$-approximate independent set (see \cite{grimmett-diarmid-1975,diarmid-1984}). Therefore, the size of the independent set found in step 3 is at least $(1/p)\log{r_Tp}(1-o(1))$, which is greater than
\[
\left(\frac{1}{p}\right)\log {(c)}(1-o(1)) =  \left(\frac{1}{p}\right)\log {(np)}(1-o(1)).
\]
Note that this set gets added to the tree obtained by the algorithm which increases the number of vertices in the tree while maintaining the acyclic property of the induced subgraph. Hence, with probability at least $3/4$, the induced subgraph has $\sum_{i=0}^{T-1}l_{i} + (1/p)\log{np}(1-o(1))$ vertices. Consequently, the FVS obtained has size at most $n-(1/p)\log{np}(1-o(1))$ with probability at least $3/4$.

\end{proof}

\section{Planted Feedback Vertex Set Problem}
We prove Theorems \ref{theorem:planted-directedFVS} and \ref{theorem:algorithm-planted-directedFVS} in this section.

The proof of Theorems \ref{theorem:planted-directedFVS} and \ref{theorem:algorithm-planted-directedFVS} are based on the following fact formalized in Lemma \ref{lemma:cycle-through-every-vertex}: if $S\subseteq V\setminus P$ is a subset of vertices of size at least $(1-\delta)n/10$, then with high probability, every vertex $u\in P$ induces a $k$-cycle with vertices in $S$. Consequently, a small hitting set $H$ for the $k$-cycles should contain either all vertices in $P$ or most vertices from $V\setminus P$. If some vertex $u\in P$ is not present in $H$, then the size of $H$ will be large since it should contain most vertices from $V\setminus P$. This contradicts the fact that $H$ is a small hitting set. Thus $H$ should contain the planted feedback vertex set $P$. This fact is stated in a general form based on the size of $H$ in Lemma \ref{lemma:small-hitting-set-contains-P}.

For Theorem \ref{theorem:planted-directedFVS}, $H$ is the smallest hitting set. By the previous argument $H\supseteq P$, and we are done since no additional vertex $v\in V\setminus P$ will be present in $H$ (in fact, $P$ is a hitting set for all cycles since it is a feedback vertex set). We formalize these arguments in this section.

\begin{lemma}\label{lemma:cycle-through-every-vertex}
Let $D_{n,\delta,p}$ be a planted directed random graph where $p\geq C/n^{1-2/k}$ for some constants $C,k,\delta$. Then, with high probability, for every vertex $v\in P$, there exists a cycle of size $k$ through $v$ in the subgraph induced by $S\cup \{v\}$ in $D_{n,\delta,p}$ if $S$ is a subset of $V\setminus P$ of size at least $|V\setminus P|/10=(1-\delta)n/10$.
\end{lemma}

We give a proof of this lemma by the second moment method later. It leads to the following important consequence which will be used to prove Theorems \ref{theorem:planted-directedFVS} and \ref{theorem:algorithm-planted-directedFVS}. It states that every sufficiently small hitting set for the $k$-cycles in $D_{n,\delta,p}$ should contain every vertex from the planted feedback vertex set.
\begin{lemma}\label{lemma:small-hitting-set-contains-P}
Let $H$ be a hitting set for the $k$-cycles in $D_{n,\delta,p}$ where $p\geq C/n^{1-2/k}$ for some constants $C,k,\delta$. If $|H|\leq t\delta n$ where $t\leq 9(1-\delta)/10\delta$, then $H\supseteq P$.
\end{lemma}

\begin{proof}
Suppose $u\in P$ and $u\not \in H$. Then $H$ should contain at least $|V\setminus P|-|V\setminus P|/10$ vertices from $V\setminus P$, else  by Lemma \ref{lemma:cycle-through-every-vertex}, there exists a $k$-cycle involving $u$ and some $k-1$ vertices among the $|V\setminus P|/10$ vertices that $H$ does not contain contradicting the fact that $H$ hits all cycles of length $k$. Therefore, $|H|>|V\setminus P|-|V\setminus P|/10=(1-\delta)9n/10 \geq t\delta n$ by the choice of $t$. Thus, the size of $H$ is greater than $t\delta n$, a contradiction.
\end{proof}

\begin{proof}[Proof of Theorem \ref{theorem:planted-directedFVS}]
We will first show that the smallest hitting set for the $k$-cycles in $D_{n,\delta,p}$ is of size exactly $|P|=\delta n$.
By Lemma \ref{lemma:cycle-through-every-vertex} there exists a $k$-cycle through every vertex $v\in P$ and some $\{u_1,\cdots,u_{k-1}\}\subset S$ if $S\subset V\setminus P$ and $|S|\geq (1-\delta)n/10$.
\begin{lemma}\label{lemma:lower-bound-HS}
If a subset $H\subseteq V$ hits all cycles of length $k$ in $D_{n,\delta,p}$, then $|H|\geq |P|$.
\end{lemma}
\begin{proof}[Proof of Lemma \ref{lemma:lower-bound-HS}]
If $H$ contains all vertices in $P$, then we are done. Suppose not. Let $u\in P$ and $u\not \in H$. Then $H$ should contain at least $|V\setminus P|-|V\setminus P|/10$ vertices from $V\setminus P$, else by Lemma \ref{lemma:cycle-through-every-vertex}, there exists a $k$-cycle involving $u$ and some $k-1$ vertices among the $|V\setminus P|/10$ vertices that $H$ does not contain. This would contradict the fact that $H$ hits all cycles of length $k$. Therefore, $|H|>|V\setminus P|-|V\setminus P|/10=(1-\delta)9n/10\geq \delta n=|P|$ since $\delta\leq 9/19$.
\end{proof}

Therefore, every hitting set for the subset of $k$-cycles should be of size at least $|P|=\delta n$. Also, we know that $P$ is a hitting set for the $k$-cycles since $P$ is a feedback vertex set in $D_{n,\delta,p}$. Thus, the optimum hitting set for the $k$-cycles is of size exactly $|P|$.

Let $H$ be the smallest hitting set for the $k$-cycles. Then $|H|=\delta n$. It is easily verified that $t=1$ satisfies the conditions of Lemma \ref{lemma:small-hitting-set-contains-P} if $\delta\leq 9/19$. Therefore, $H\supseteq P$. Along with the fact that $H=\delta n =|P|$, we conclude that $H=P$.
\end{proof}

\subsection{Algorithm to Recover Planted Feedback Vertex Set}
In this section, we give an algorithm to recover the planted feedback vertex set in $D_{n,\delta,p}$ thereby proving Theorem \ref{theorem:algorithm-planted-directedFVS}. Theorem \ref{theorem:planted-directedFVS} suggests an algorithm where one would obtain all cycles of length $k$ and find the best hitting set for these set of cycles. Even though the number of $k$-cycles is polynomial, we do not have a procedure to find the best hitting set for $k$-cycles. However, by repeatedly taking all vertices of a cycle into the hitting set and removing them from the graph, we do have a simple greedy strategy that finds a $k$-approximate hitting set. We will use this strategy to give an algorithm that recovers the planted feedback vertex set.

\noindent {\bf Algorithm Recover-Planted-FVS($D_{n,\delta,p}=D(V,E)$)}
\begin{enumerate}
\item Obtain cycles in increasing order of size until all cycles of length $k$ are obtained. Let $\T'$ be the subset of cycles. Let $S$ be the empty set.
\item While there exists a cycle $T\in \T'$ such that $S$ does not hit $T$,
\begin{enumerate}
\item Add all vertices in $T$ to $S$.
\end{enumerate}
\item Return $H$, where $H=\{u\in S:\exists$ $k$-cycle through $v$ in the subgraph induced by $V\setminus S\cup\{u\}\}$.
\end{enumerate}

The idea behind the algorithm is the following: The set $S$ obtained at the end of step $2$ in the above algorithm is a $k$-approximate hitting set and hence is of size at most $k\delta n$. Using Lemma \ref{lemma:small-hitting-set-contains-P}, it is clear that $S$ contains $P$ - indeed, if $S$ does not contain all vertices in $P$, then $S$ should contain most of the vertices in $V\setminus P$ contradicting the fact that the size of $S$ is at most $k\delta n$. Further, owing to the choice of $\delta$, it can be shown that $S$ does not contain at least $|V\setminus P|/10$ vertices from $V\setminus P$.  Therefore, by Lemma \ref{lemma:cycle-through-every-vertex}, every vertex $v\in P$ induces a $k$-cycle with some subset of vertices from $V\setminus S$. Also, since $V\setminus P$ is a DAG no vertex $v\in V\setminus P$ induces cycles with any subset of vertices from $V\setminus S\subseteq V\setminus P$. Consequently, a vertex $v$ induces a $k$-cycle with vertices in $V\setminus S$ if and only if $v\in P$. Thus, the vertices in $P$ are identified exactly.

\begin{proof}[Proof of Theorem \ref{theorem:algorithm-planted-directedFVS}]
We use Algorithm Recover-Planted-FVS to recover the planted feedback vertex set from the given graph $D=D_{n,\delta,p}$. Since we are using the greedy strategy to obtain a hitting set $S$ for $\T'$, it is clear that $S$ is a $k$-approximate hitting set. Therefore $|S|\leq k\delta n$. It is easily verified that $t=k$ satisfies the conditions of Lemma \ref{lemma:small-hitting-set-contains-P} if $\delta\leq 1/2k$. Thus, all vertices from the planted feedback vertex set $P$ are present in the subset $S$ obtained at the end of step 2 in the algorithm.

By the choice of $\delta\leq 1/2k$, it is true that $|S|\leq k\delta n\leq 9(1-\delta)n/10=9|V\setminus P|/10$. Hence, $|V\setminus S|\geq |V\setminus P|/10$.

Since $S\supseteq P$, the subset of vertices $V\setminus S$ does not contain any vertices from the planted set. Also, the number of vertices in $V\setminus S$ is at least $|V\setminus P|/10$. Consequently, by Lemma \ref{lemma:cycle-through-every-vertex}, each vertex $v\in P$ induces at least one $k$-cycle with vertices in $V\setminus S$. Since $V\setminus P$ is a DAG, none of the vertices $u\in V\setminus P$ induce cycles with vertices in $V\setminus S$. Therefore, a vertex $v\in S$ induces a $k$-cycle with vertices in $V\setminus S$ if and only if $v\in P$. Hence, the subset $H$ output by Algorithm Recover-Planted-FVS is exactly the planted feedback vertex set $P$.

Next we prove that the algorithm runs in polynomial time in expectation. The following lemma shows an upper bound on the expected number of cycles of length $k$. It is proved later by a simple counting argument.

\begin{lemma}\label{lemma:expected-no-of-cycles}
The expected number of cycles of length $k$ in $D_{n,\delta,p}$ is at most $(nkp)^{k}$.
\end{lemma}
Since the expected number of cycles obtained by the algorithm is $(nkp)^{k}$ by Lemma \ref{lemma:expected-no-of-cycles}, the algorithm uses $(nkp)^k$-sized storage memory. Finally, since the size of $\T'$ is $(nkp)^k$, steps 2 and 3 of the algorithm can be implemented to run in expected $(nkp)^{O(k)}$ time.
\end{proof}

\section{Proofs}
\subsection{Lower Bound for FVS in Random Graphs}
In this section, we prove the lower bound for the Feedback Vertex Set in random graphs. We consider the dual problem - namely the maximum induced acyclic subgraph.

We will need the following bound on the number of ways to partition a positive integer $n$ into $k$ positive integers.
\begin{theorem}\label{theorem:partition-function}\cite{wladimir-partition-function}
Let $p_k(n)$ denote the number of ways to partition $n$ into exactly $k$ parts. Then there exists an absolute constant $A<1$ such that
\[
p_k(n)<A\frac{e^{c\sqrt{n-k}}}{(n-k)^{3/4}}e^{\frac{-2\sqrt{n-k}}{c}}L_2(e^{-\frac{c(k+1/2)}{2\sqrt{n-k}}})
\]
where $c=\pi\sqrt{2/3}$ and $L_2(x)=\sum_{m=1}^{\infty}\frac{x^m}{m^2}$ for $|x|\leq 1$.
\end{theorem}
\noindent {Remark 3.} Since we will not need such a tight bound, we will use $p_k(n)< C_1e^{C_2(n-k)}$ for some constants $C_1,C_2>0$.

We prove Theorem \ref{theorem:undirected-FVSlowerbound} now based on simple counting arguments. We observe that the proof of Theorem \ref{theorem:directed-FVSlowerbound} given by Spencer and Subramanian is also based on similar counting arguments while observing that if a directed graph is acyclic, then there exists an ordering of the vertices such that each arc is in the forward direction.
\begin{proof}[Proof of Theorem \ref{theorem:undirected-FVSlowerbound}]
First note that every induced subgraph on $r$ vertices is a graph from the family $G(r,p)$. We bound the probability that a graph $H=G(r,p)$ is a forest.\\

\begin{align*}
\prob{H\ \text{is a forest}}&\leq \sum_{k=1}^r \sum_{n_1+\cdots+n_k=r, n_i>0} \text{No. of forests with spanning}\\
&\ \ \ \ \ \text{trees on $n_1,\cdots,n_k$ vertices}\\
&\ \ \ \ \ \ \  \times \prob{\text{Forest with $k$ components}}\\
&= \sum_{k=1}^r \sum_{n_1+\cdots+n_k=r, n_i>0} \left(\frac{r!}{\prod_{i=1}^k n_i!}\right) \left(\prod_{i=1}^k n_i^{n_i-2} \right)\\
&\ \ \ \ \ \ \ \times p^{r-k}(1-p)^{\binom{r}{2}-r+k}\\
&\leq r!(1-p)^{\binom{r}{2}} \sum_{k=1}^r \sum_{n_1+\cdots+n_k=r, n_i>0} \left(\frac{p}{1-p}\right)^{r-k}\\
&\leq r!(1-p)^{\binom{r}{2}} \sum_{k=1}^r \sum_{n_1+\cdots+n_k=r, n_i>0} \left(2p\right)^{r-k} \\
& \ \ \ \quad \quad \text{(since $p<1/2$)}\\
&\leq r!(1-p)^{\binom{r}{2}} \sum_{k=1}^r \left(2p\right)^{r-k} \sum_{n_1+\cdots+n_k=r, n_i>0} 1\\
&= r!(1-p)^{\binom{r}{2}} \sum_{k=1}^r \left(2p\right)^{r-k} p_k(r)\\
&\leq r!(1-p)^{\binom{r}{2}} \sum_{k=1}^r \left(2p\right)^{r-k} C_1 e^{C_2(r-k)} \\
& \ \ \ \quad \quad \text{(by Remark 3)}\\
&\leq C_1r^r(1-p)^{\binom{r}{2}} \sum_{k=1}^r (2e^{C_2}p)^{r-k}
\end{align*}
\begin{align*}
&\leq C_1(1-p)^{\frac{r^2}{2}}n^r \sum_{k=1}^r(2e^{C_2}p)^{r-k} \\
& \ \ \ \quad \quad \text{(since $r\leq n$)}\\
&\leq C_1(1-p)^{\frac{r^2}{2}}r(2e^{C_2}np)^r \\
& \leq e^{-r\left(\frac{pr}{2}-\log{(2e^{C_2}np)}-\frac{\log{(C_1r)}}{r}\right)}\\
\end{align*}
which tends to zero when $r>\frac{2}{p}(\log{np})(1+o(1))$.
\end{proof}

\subsection{Feedback Vertex Set in Random Graphs}
We will use the following Chernoff bound for the concentration of the binomial distribution.
\begin{lemma} \label{lemma:chernoff}
Let $X=\sum_{i=1}^n X_i$ where $X_i$ are i.i.d. Bernoulli random variables with $\prob{X_i=1}=p$. Then
\[
\prob{|X-np|\geq a\sqrt{np}}\leq 2e^{-a^2/2}.
\]
\end{lemma}

\begin{proof}[Proof of Lemma \ref{lemma:largesurvivors}]
We prove the lemma by induction on $t$. We will prove the stronger induction hypothesis that every $l_i$, $u_i$ for $i\in\{0,1,\cdots,t\}$ satisfy their respective concentration bounds with probability at least
\[
a_t:= 1-\frac{t}{16T} -\frac{1}{16}\sum_{i=1}^{t}1/i^2.
\]
We will prove the concentration of $r_{i+1}$ as a consequence of $l_i$ and $u_i$ satisfying their respective concentration bounds. We will in fact show that the failure probability of $r_{i+1}$ satisfying its concentration bound conditioned on $l_i$ and $u_i$ satisfying their respective concentration bounds will be at most $1/(32(i+1)^2)$. It immediately follows that with failure probability at most $(t/16T)+(3/32)\sum_{i=1}^{t}(1/i^2)+(1/32(t+1)^2)\leq 1/4$, every $r_{i+1}$, $u_i$ and $l_i$, for $i\in\{0,1,\cdots,t\}$ satisfies its respective concentration bound leading to the conclusion of the lemma.

For the base case, consider $t=0$. It is clear that $u_0=n-1$ and $l_0=1$ satisfy the concentration bounds with probability $1$. For the induction step, the induction hypothesis is the following: With probability at least $a_t$, the concentration bounds are satisfied for $u_i$ and $l_i$ for every $i\in \{0,1,\cdots,t\}$. We will bound the probability that $u_{t+1}$ or $l_{t+1}$ fails to satisfy its corresponding concentration bound conditioned on the event that $u_i,l_i$ for $i\in\{0,1,\cdots,t\}$ satisfy their respective concentration bounds.

\noindent 1. To prove the concentration bound for $u_{t+1}$, observe that $u_{t+1}$ is a binomial distribution with $u_t$ trials and success probability $(1-p)^{l_t}$. Indeed, $u_{t+1}$ is the number of vertices among $U_t$ which are not neighbors of vertices in $L_t$. For each vertex $x\in U_t$, $\prob{\text{$x$ has no neighbor in $L_t$}}=(1-p)^{l_t}$.

Therefore, by Lemma \ref{lemma:chernoff}, we have that
$\prob{|u_{t+1}-u_t(1-p)^{l_t}|>\gamma_{t+1}\sqrt{u_t(1-p)^{l_t}}}$
\[
\leq 2e^{-\gamma_{t+1}^2/2}=\frac{1}{32(t+1)^2}
\]
with $\gamma_{t+1}=\sqrt{4\ln{8(t+1)}}$. Hence, with probability at least $1-(1/32(t+1)^2)$,
\begin{align*}
u_{t+1} &\leq u_t(1-p)^{l_t}\left(1+\sqrt{\frac{4\ln{8(t+1)}}{u_t(1-p)^{l_t}}}\right),\\
u_{t+1} &\geq u_t(1-p)^{l_t}\left(1-\sqrt{\frac{4\ln{8(t+1)}}{u_t(1-p)^{l_t}}}\right).
\end{align*}

Now, using the bounds on $u_t$ and $l_t$, 
\begin{align*}
\frac{4\ln{8(t+1)}}{u_t(1-p)^{l_t}} & \leq \frac{10\ln{\ln{n}}}{n}
\end{align*}
since $t+1\leq T\leq \ln{n}$,
\begin{align*}
(n-\sum_{i=0}^t (c+20\sqrt{c})^i) &\geq \frac{15n}{16},\\
(1-p(c+20\sqrt{c})^t) &\geq \frac{15}{16} \quad \text{and}\\
\left(1-\sqrt{\frac{\ln{\ln{n}}}{n}}\right) &\geq \frac{1}{2}.
\end{align*}

Hence, 
\begin{align}
u_{t+1} &\leq u_t(1-p)^{l_t}\left(1+\sqrt{\frac{\ln{\ln{n}}}{n}}\right) \label{ineq:u-upperbound}\\
u_{t+1} &\geq u_t(1-p)^{l_t}\left(1-\sqrt{\frac{\ln{\ln{n}}}{n}}\right) \label{ineq:u-lowerbound}.
\end{align}

Therefore, 
\begin{align*}
u_{t+1} &\geq u_t(1-p)^{l_t}\left(1-\sqrt{\frac{\ln{\ln{n}}}{n}}\right)\\
&\ \ \ \quad \quad \text{(Using inequality \ref{ineq:u-lowerbound})}\\
&\geq u_t(1-l_tp)\left(1-\sqrt{\frac{\ln{\ln{n}}}{n}}\right)\\
&\geq (n-\sum_{i=0}^t (c+20\sqrt{c})^i)\left(1-\frac{c(c+20\sqrt{c})^t}{n}\right)\\
&\ \ \times \left(1-\sqrt{\frac{\ln{\ln{n}}}{n}}\right)\\
&\quad \quad \text{(Using the bounds on $u_t$ and $l_t$)}\\
&\geq (n-\sum_{i=0}^t (c+20\sqrt{c})^i)\left(1-\frac{(c+20\sqrt{c})^{t+1}}{n}\right)\\
&\ \ \times \left(1-\sqrt{\frac{\ln{\ln{n}}}{n}}\right)\\
&= \left(n-\sum_{i=0}^t (c+20\sqrt{c})^i - (c+20\sqrt{c})^{t+1}\right. \\
&\ \ + \left.\frac{(c+20\sqrt{c})^{t+1}}{n}\sum_{i=0}^t (c+20\sqrt{c})^i\right)\\
&\ \ \ \ \times \left(1-\sqrt{\frac{\ln{\ln{n}}}{n}}\right)
\end{align*}
\begin{align*}
&\geq \left(n- \sum_{i=0}^{t+1} (c+20\sqrt{c})^i\right)\left(1-\sqrt{\frac{\ln{\ln{n}}}{n}}\right)
\end{align*}
which proves the lower bound. The upper bound is obtained by proceeding similarly:
\begin{align*}
u_{t+1} &\leq u_t(1-p)^{l_t}\left(1+\sqrt{\frac{\ln{\ln{n}}}{n}}\right) \\
& \ \ \quad \quad \text{(Using inequality \ref{ineq:u-upperbound})}\\
&\leq u_t\left(1-\frac{l_tp}{2}\right)\left(1+\sqrt{\frac{\ln{\ln{n}}}{n}}\right)\\
&\leq u_t\left(1-\frac{c(c-20\sqrt{c})^t}{n}(1-16Tp(c+20\sqrt{c})^t)\right.\\
&\ \ \ \left.\left(1-\frac{\sum_{i=0}^t(c+20\sqrt{c})^i}{n}\right)\right)\left(1+\sqrt{\frac{\ln{\ln{n}}}{n}}\right)\\
&\quad \quad \text{(Using the bound on $l_t$)}\\
&\leq u_t\left(1-\frac{c(c-20\sqrt{c})^t}{4n}\right)\left(1+\sqrt{\frac{\ln{\ln{n}}}{n}}\right) \\
& \quad \quad \left(\text{Since $(1-16Tp(c+20\sqrt{c})^t)\geq \frac{1}{2}$,}\right.\\
&\ \ \quad \quad \left.\text{$\left(1-\frac{\sum_{i=0}^t(c+20\sqrt{c})^i}{n}\right)\geq \frac{15}{16}$}\right)\\
&\leq \left(n-\frac{\sum_{i=0}^t(c-20\sqrt{c})^i}{4n}\right)\left(1-\frac{c(c-20\sqrt{c})^t}{4n}\right)\\
&\ \ \ \times \left(1+\sqrt{\frac{\ln{\ln{n}}}{n}}\right)\\
&\leq \left(n-\frac{\sum_{i=0}^t(c-20\sqrt{c})^i}{4n}\right)\left(1-\frac{(c-20\sqrt{c})^{t+1}}{4n}\right)\\
&\ \ \ \times \left(1+\sqrt{\frac{\ln{\ln{n}}}{n}}\right)\\
&\leq \left(n-\frac{\sum_{i=0}^{t+1}(c-20\sqrt{c})^i}{4n}\right)\left(1+\sqrt{\frac{\ln{\ln{n}}}{n}}\right).
\end{align*}

Thus, $u_{t+1}$ satisfies the concentration bound with failure probability at most $1/(32(t+1)^2)$ conditioned on the event that $u_i,l_i$ for $i\in\{0,1,\cdots,t\}$ satisfy their respective concentration bounds.

\noindent 2. Next we address the failure probability of $r_{t+1}$ not satisfying its concentration bound conditioned on the event that $u_i,l_i$ for $i\in\{0,1,\cdots,t\}$ satisfy their respective concentration bounds. Lemma \ref{lemma:lowerbound-R} proves that the number of {\bf \emph{unique}} neighbors $r_{t+1}$ is concentrated around its expectation.

\begin{lemma}\label{lemma:lowerbound-R}
Let $q_t:=pl_t(1-p)^{l_t-1}$. With probability at least $1-(1/32(t+1)^2)$ 
\begin{align*}
q_tu_t\left(1+\frac{20}{\sqrt{c}}\right) \geq r_{t+1} \geq q_tu_t\left(1-\frac{20}{\sqrt{c}}\right)\\
\end{align*}
when $t+1\leq T$.
\end{lemma}
\begin{proof}[Proof of Lemma \ref{lemma:lowerbound-R}]
Observe that $r_{t+1}$ is a binomially distributed random variable with $u_t$ trials and success probability $q_t$. Indeed, $r_{t+1}$ is the number of vertices among $U_t$ which are adjacent to exactly one vertex in $L_t$. For each $u\in U_t$, $\prob{\text{$u$ is adjacent to exactly one vertex in $L_t$} }=pl_t(1-p)^{l_t-1}=q_t$.

Using $\beta_{t+1}=\sqrt{4\ln{8(t+1)}}$, by Lemma \ref{lemma:chernoff}, we have that
$\prob{|r_{t+1}-q_tu_t|>\beta_{t+1}\sqrt{q_tu_t}}$
\begin{align*}
&\leq 2e^{-\beta_{t+1}^2/2}=\frac{1}{32(t+1)^2}.
\end{align*}
Hence, with probability at least $1-(1/32(t+1)^2)$,
\begin{align}
r_{t+1} &\leq q_tu_t\left(1+\sqrt{\frac{4\ln{8(t+1)}}{q_tu_t}}\right) \label{ineq:r-lowerbound}\\
r_{t+1} &\geq \geq q_tu_t\left(1-\sqrt{\frac{4\ln{8(t+1)}}{q_tu_t}}\right).  \label{ineq:r-upperbound}
\end{align}

Lemma \ref{lemma:lowerbound-qu} proves the concentration of the expected number of unique neighbors of $L_t$ conditioned on the event that $u_i,l_i$ for $i\in\{0,1,\cdots,t\}$ satisfy their respective concentration bounds. This in turn helps in proving that $r_{t+1}$ is concentrated.
\begin{lemma}\label{lemma:lowerbound-qu}
For $t+1\leq T$, if $u_t$ and $l_t$ satisfy their respective concentration bounds, then
\begin{enumerate}
\item $q_tu_t \leq c(c+20\sqrt{c})^{t}\left(1+\sqrt{\frac{\ln{\ln{n}}}{n}}\right)$,
\item $q_tu_t$\\
$\geq \frac{c(c-20\sqrt{c})^{t}}{4} \left(1-\frac{\sum_{i=0}^{t+1}(c+20\sqrt{c})^i}{n}\right)\left(1-\sqrt{\frac{\ln{\ln{n}}}{n}}\right)$.
\end{enumerate}
\end{lemma}
\begin{proof}[Proof of Lemma \ref{lemma:lowerbound-qu}]
Recall that $q_t=pl_t(1-p)^{l_t-1}$. Hence, 
\begin{align*}
q_tu_t &\geq p(n-\sum_{i=0}^t(c+20\sqrt{c})^i)l_t(1-p)^{l_t-1}\\
&\ \ \times\left(1-\sqrt{\frac{\ln{\ln{n}}}{n}}\right)\\
 &= pn\left(1-\frac{\sum_{i=0}^t(c+20\sqrt{c})^i}{n}\right)l_t(1-p)^{l_t-1}\\
 &\ \ \ \times \left(1-\sqrt{\frac{\ln{\ln{n}}}{n}}\right)\\
 &\geq c\left(1-\frac{\sum_{i=0}^t(c+20\sqrt{c})^i}{n}\right)l_t(1-l_tp)\\
 &\ \ \ \times \left(1-\sqrt{\frac{\ln{\ln{n}}}{n}}\right)\\
 &\geq c(c-20\sqrt{c})^{t}\left(1-\frac{\sum_{i=0}^t(c+20\sqrt{c})^i}{n}\right)^2\\
 &\ \ \ \times(1-16Tp(c+20\sqrt{c})^t)(1-p(c+20\sqrt{c})^t)\\
 &\ \ \ \times\left(1-\sqrt{\frac{\ln{\ln{n}}}{n}}\right)\\
&\quad \quad \text{(By the bound on $l_t$)}\\
&\geq \frac{c(c-20\sqrt{c})^{t}}{4} \left(1-\frac{\sum_{i=0}^{t+1}(c+20\sqrt{c})^i}{n}\right)\\
 &\ \ \ \times\left(1-\sqrt{\frac{\ln{\ln{n}}}{n}}\right)
\end{align*}
using Lemma \ref{lemma:handler1} and
\begin{align*}
(1-16Tp(c+20\sqrt{c})^t) &\geq \frac{1}{2},\\
(1-p(c+20\sqrt{c})^t) &\geq \frac{1}{2} \quad \quad \text{when $t+1\leq T$.}
\end{align*}
For the upper bound:
\begin{align*}
q_tu_t &=pl_t(1-p)^{l_t-1}u_t\\
&\leq pl_tu_t\\
&\leq pl_t\left(n-\frac{\sum_{i=0}^t(c-20\sqrt{c})^i}{4}\right)\left(1+\sqrt{\frac{\ln{\ln{n}}}{n}}\right)\\
&\ \ \ \ \quad \quad \text{(Using the bound on $u_t$)}\\
&\leq cl_t\left(1-\frac{\sum_{i=0}^t(c-20\sqrt{c})^i}{4n}\right)\left(1+\sqrt{\frac{\ln{\ln{n}}}{n}}\right)\\
\end{align*}
\begin{align*}
&\leq c(c+20\sqrt{c})^t\left(1-\frac{\sum_{i=0}^t(c-20\sqrt{c})^i}{4n}\right)\\
&\ \ \ \times \left(1+\sqrt{\frac{\ln{\ln{n}}}{n}}\right)\\
&\ \ \ \ \quad \quad \text{(Using the bound on $l_t$)}\\
&\leq c(c+20\sqrt{c})^t\left(1+\sqrt{\frac{\ln{\ln{n}}}{n}}\right) \\
&\ \quad \quad \left(\text{Since $\left(1-\frac{\sum_{i=0}^t(c-20\sqrt{c})^i}{4n}\right)\leq 1$}\right).
\end{align*}
\end{proof}

Consequently, using Lemma \ref{lemma:lowerbound-qu}, 
\begin{align*}
\frac{4\ln{8(t+1)}}{q_tu_t} &\leq \frac{400}{c}
\end{align*}
since, when $t+1\leq T$,
\begin{align*}
\left(1-\frac{\sum_{i=0}^{t+1}(c+20\sqrt{c})^i}{n}\right) &\geq \left(\frac{15}{16}\right)^2,\\
\left(1-\sqrt{\frac{\ln{\ln{n}}}{n}}\right) &\geq \frac{1}{2} \quad \text{and}\\
\frac{1}{2} &\geq \frac{4\ln{8(t+1)}}{(c-20\sqrt{c})^{t}}.
\end{align*}

Hence, by inequalities \ref{ineq:r-lowerbound} and \ref{ineq:r-upperbound}, with probability at least $1-(1/32(t+1)^2)$, 
\begin{align}
q_tu_t\left(1+\frac{20}{\sqrt{c}}\right) \geq r_{t+1} \geq q_tu_t\left(1-\frac{20}{\sqrt{c}}\right)
\end{align}
when $t+1\leq T$. This concludes the proof of Lemma \ref{lemma:lowerbound-R}
\end{proof}
Lemmas \ref{lemma:lowerbound-R} and \ref{lemma:lowerbound-qu} together show that $r_{t+1}$ satisfies the concentration bounds with failure probability at most $(1/32(t+1)^2)$ conditioned on the event that $u_t$ and $l_t$ satisfy their respective concentration bounds.

\noindent 3. Finally we address the failure probability of $l_{t+1}$ satisfying its concentration bound conditioned on the event that $u_i,l_i$ for $i\in\{0,1,\cdots,t\}$ satisfy their respective concentration bounds. By Step 2(e) of the algorithm, the number of surviving vertices in level $t+1$ is $l_{t+1}:=r_{t+1} - m_{t+1}$, where $m_{t+1}$ denotes the number of edges among the vertices in $R_{t+1}$. In Lemma \ref{lemma:lowerbound-R}, we showed that the number of {\bf \emph{unique}} neighbors $r_{t+1}$ is concentrated around its expectation. Lemma \ref{lemma:upperbound-M} proves a concentration which bounds the number of edges among the vertices in $R_t$. These two bounds will immediately lead to the induction step on $l_{t+1}$. Thus, the probability that $l_{t+1}$ does not satisfy its concentration bound will at most be the probability that either $m_{t+1}$ or $r_{t+1}$ does not satisfy its respective concentration bound.

\begin{lemma}\label{lemma:upperbound-M}
$m_{t+1}\leq 8Tr_{t+1}^2p$ with probability at least $1-(1/16T)$.
\end{lemma}

\begin{proof}[Proof of Lemma \ref{lemma:upperbound-M}]
Recall that $m_{t+1}$ denotes the number of edges among the vertices in $R_{t+1}$. Since the algorithm has not explored the edges among the vertices in $R_{t+1}$, $m_{t+1}$ is a random variable following the Binomial distribution with $\binom{r_{t+1}}{2}$ trials and success probability $p$. By Markov's inequality, we have that for $t+1\leq T$,
\[
\prob{m_{t+1}\geq 8Tr_{t+1}^2p}\leq \frac{1}{16T}.
\]
Hence, $m_{t+1}\leq 8Tr_{t+1}^2p$ with probability at least $1-(1/16T)$, .
\end{proof}

Recollect that $l_{t+1}=r_{t+1}-m_{t+1}$. The upper bound of the induction step follows using Lemma \ref{lemma:lowerbound-qu}:
\begin{align*}
l_{t+1} &\leq r_{t+1}\\
&\leq q_tu_t\left(1+\frac{20}{\sqrt{c}}\right)\\
&\leq c(c+20\sqrt{c})^{t}\left(1+\sqrt{\frac{\ln{\ln{n}}}{n}}\right)\left(1+\frac{20}{\sqrt{c}}\right)\\
&\leq (c+20\sqrt{c})^{t+1} \left(1+\sqrt{\frac{\ln{\ln{n}}}{n}}\right).
\end{align*}

For the lower bound, we use Lemmas \ref{lemma:lowerbound-R} and \ref{lemma:upperbound-M} conditioned on the event that $l_t$ and $u_t$ satisfy their respective concentration bounds. With failure probability at most
\[
\frac{1}{32(t+1)^2}+\frac{1}{16T},
\]
we have that $l_{t+1}$
\begin{align*}
&=r_{t+1}-m_{t+1}\\
&\geq r_{t+1} - 8Tr_{t+1}^2p\\
&= r_{t+1}(1-8Tr_{t+1}p)\\
&\geq q_tu_t\left(1-\frac{20}{\sqrt{c}}\right)\left(1-8Tq_tu_tp\left(1+\frac{20}{\sqrt{c}}\right)\right)\\
&\quad \quad \text{(Using Lemma \ref{lemma:lowerbound-R})}\\
&= l_tp(1-p)^{l_t-1}u_t\left(1-8Tl_tp^2(1-p)^{l_t-1}u_t\right.\\
&\ \ \left.\times\left(1+\frac{20}{\sqrt{c}}\right)\right)\left(1-\frac{20}{\sqrt{c}}\right)\\
& \quad \quad \text{(Substituting for $q_t=pl_t(1-p)^{l_t-1}$)}\\
& \geq l_tp\left(1-\frac{20}{\sqrt{c}}\right)(1-l_tp)(1-12Tl_tp^2(1-p)^{l_t-1}u_t)\\
\end{align*}
\begin{align*}
&\geq l_tp\left(1-\frac{20}{\sqrt{c}}\right)(1-l_tp)(1-12Tnp^2l_t(1-p)^{l_t-1})\\
& \quad \quad \text{(Since $u_t\leq n$)}\\
&\geq l_tp\left(1-\frac{20}{\sqrt{c}}\right)(1-l_tp)(1-12Tcpl_t(1-p)^{l_t-1})\\
&\geq l_tp\left(1-\frac{20}{\sqrt{c}}\right)(1-l_tp-12Tcpl_t(1-p)^{l_t}(1-l_tp))\\
&\geq l_tp\left(1-\frac{20}{\sqrt{c}}\right)(1-l_tp(1+12Tc))\\
&\geq l_tpu_t(1-l_tp(1+12Tc))\left(1-\frac{20}{\sqrt{c}}\right)\\
&\geq l_tp(n-\sum_{i=0}^t {(c+20\sqrt{c})^i})(1-l_tp(1+12Tc))\\
&\ \ \ \times \left(1-\frac{20}{\sqrt{c}}\right) 
 \quad \quad \text{(Using the bound on $u_t$)}\\
&\geq l_tnp\left(1-\frac{\sum_{i=0}^t {(c+20\sqrt{c})^i}}{n}\right)(1-l_tp(1+12Tc))\\
&\ \ \ \times \left(1-\frac{20}{\sqrt{c}}\right)\\
&= l_tc\left(1-\frac{\sum_{i=0}^t {(c+20\sqrt{c})^i}}{n}\right)(1-l_tp(1+12Tc))\\
&\ \ \ \times\left(1-\frac{20}{\sqrt{c}}\right)\\
&\geq c(c-20\sqrt{c})^{t}\left(1-\frac{\sum_{i=0}^t {(c+20\sqrt{c})^i}}{n}\right)^2\\
&\ \ \ \times(1-16Tp(c+20\sqrt{c})^t)\\
&\quad \quad \times (1-(c+20\sqrt{c})^tp(1+12Tc))\left(1-\frac{20}{\sqrt{c}}\right)\\
&\ \ \ \quad \text{(using the bound on $l_t$)}\\
&\geq (c-20\sqrt{c})^{t+1}\left(1-\frac{\sum_{i=0}^{t+1} {(c+20\sqrt{c})^i}}{n}\right)\\
&\ \ \ \times(1-16Tp(c+20\sqrt{c})^{t+1}) \quad \quad \text{(Using Lemma \ref{lemma:handler1})}
\end{align*}
proving the induction step of the lower bound for $l_{t+1}$.

Thus, $l_{t+1}$ satisfies the concentration bounds with failure probability at most $(1/32(t+1)^2)+(1/16T)$ conditioned on the event that $u_i,l_i$ for $i\in\{0,1,\cdots,t\}$ satisfy their respective concentration bounds.

Finally, by the union bound, with probability at most $\frac{1}{32(t+1)^2}+\frac{1}{32(t+1)^2}+\frac{1}{16T}$, either $u_{t+1}$ or $l_{t+1}$ does not satisfy its respective concentration bounds conditioned on the event that $u_i,l_i$ for $i\in\{0,1,\cdots,t\}$ satisfy their respective concentration bounds. By induction hypothesis, the failure probability of some $u_i,l_i$ for $i\in\{0,1,\cdots,t\}$ not satisfying their respective concentration bound is at most $1-a_t$. Hence, the probability that  $u_i,l_i$ satisfy their respective concentration bound for every $i\in \{0,1,\cdots,t+1\}$ is at least $a_t(1-(1/16(t+1)^2)-(1/16T))\geq a_{t+1}$. Therefore, with probability at least $a_{t+1}$, every $u_i,l_i$ for $i\in\{0,1,\cdots,t+1\}$ satisfy their respective concentration bounds. This proves the stronger induction hypothesis.

To complete the proof of Lemma \ref{lemma:largesurvivors}, recollect that we showed that the failure probability of $r_{i+1}$ satisfying its concentration bound conditioned on $l_i$ and $u_i$ satisfying their respective concentration bounds is at most $1/(32(i+1)^2)$. By the union bound argument, it immediately follows that with failure probability at most $(t/16T)+(3/32)\sum_{i=1}^{t}(1/i^2)+(1/32(t+1)^2)\leq 1/4$, every $r_{i+1}$, $u_i$ and $l_i$, for $i\in\{0,1,\cdots,t\}$ satisfies its respective concentration bound. This concludes the proof of Lemma \ref{lemma:largesurvivors}.
\end{proof}

\begin{lemma} \label{lemma:handler1}
For $t+1\leq T$,
\begin{enumerate}
\item
\[
\left(1-\frac{\sum_{i=0}^t (c+20\sqrt{c})^i}{n}\right)^{2} \geq 1-\frac{\sum_{i=0}^{t+1} (c+20\sqrt{c})^i}{n}
\]
\item $\left(1-16Tp(c+20\sqrt{c})^t\right)(1-(c+20\sqrt{c})^tp(1+12Tc))$
\begin{align*}
&\geq \left(1-16Tp(c+20\sqrt{c})^{t+1}\right)
\end{align*}
\end{enumerate}
\end{lemma}
\begin{proof} [Proof of Lemma \ref{lemma:handler1}]
We prove the first part of the Lemma by induction. For the base case, we need to prove that
\begin{align*}
1+\frac{1}{n^2}-\frac{2}{n}&\geq 1-\frac{c+20\sqrt{c}}{n}-\frac{1}{n}\\
\text{i.e., to prove that}\quad n-1&\leq (c+20\sqrt{c})n
\end{align*}
which is true.
For the induction step, we need to prove that
\[
\left(1-\frac{\sum_{i=0}^t(c+20\sqrt{c})^i}{n}-\frac{(c+20\sqrt{c})^{t+1}}{n}\right)^2
\]
\[
\geq 1-\frac{\sum_{i=0}^{t+2}(c+20\sqrt{c})^i}{n}
\]
Now, LHS
\begin{align*}
&= \left(1-\frac{\sum_{i=0}^{t}(c+20\sqrt{c})^i}{n}\right)^2 + \frac{(c+20\sqrt{c})^{2t+2}}{n^2}\\
&\ \ \ -\frac{2(c+20\sqrt{c})^{t+1}}{n}\left(1-\frac{\sum_{i=0}^{t}(c+20\sqrt{c})^i}{n}\right)\\
&\geq 1-\frac{\sum_{i=0}^{t+1} (c+20\sqrt{c})^i}{n} + \frac{(c+20\sqrt{c})^{2t+2}}{n^2}\\
&\ \ \ -\frac{2(c+20\sqrt{c})^{t+1}}{n}+\frac{2(c+20\sqrt{c})^{t+1}\sum_{i=0}^{t}(c+20\sqrt{c})^i}{n^2}.
\end{align*}
Hence, it is sufficient to prove that
\begin{align*}
-\frac{(c+20\sqrt{c})^{t+2}}{n} &\leq \frac{(c+20\sqrt{c})^{2t+2}}{n^2}-\frac{2(c+20\sqrt{c})^{t+1}}{n}\\
&\ \ \ +\frac{2(c+20\sqrt{c})^{t+1}\sum_{i=0}^{t}(c+20\sqrt{c})^i}{n^2} \\
(c+20\sqrt{c}) &\geq 2 - \frac{(c+20\sqrt{c})^{t+1}}{n}\\
&\ \ \ - 2\frac{\sum_{i=0}^{t}(c+20\sqrt{c})^i}{n},
\end{align*}
which is true for large enough $c$ when $t+1\leq T$.

For the second part of the Lemma, we need to prove that
$\left(1-16Tp(c+20\sqrt{c})^t\right)(1-(c+20\sqrt{c})^tp(1+12Tc))$
\begin{align*}
&\geq \left(1-16Tp(c+20\sqrt{c})^{t+1}\right)
\end{align*}
i.e., $1-16Tp(c+20\sqrt{c})^t-(c+20\sqrt{c})^tp(1+12Tc)+18Tp^2(c+20\sqrt{c})^{2t}(1+12Tc)$
\begin{align*}
&\geq 1-16Tp(c+20\sqrt{c})^{t+1}
\end{align*}
i.e.,
\[
(1-16Tp(c+20\sqrt{c})^t)(1+12Tc)\leq 16T(c+20\sqrt{c}-1)
\]
which is true since $1+12Tc\leq 16T(c+20\sqrt{c}-1)$ for large $c$ and the rest of the terms are less than $1$ when $t+1\leq T$.

\end{proof}

\subsection{Planted Feedback Vertex Set}
We prove Lemma \ref{lemma:cycle-through-every-vertex} by the second moment method.
\begin{proof}[Proof of Lemma \ref{lemma:cycle-through-every-vertex}]
Let $S\subset V\setminus P$, $|S|\geq (1-\delta)n/10$, $v\in P$. Let $X_v$ denote the number of cycles of size $k$ through $v$ in the subgraph induced by $S\cup\{v\}$. Then, $\E(X_v)=\binom{(1-\delta)n/10}{k-1}p^k$. Using Chebyshev's inequality, we can derive that
\[
\prob{X_v=0}\leq \frac{\var{X_v}}{\E(X_v)^2}.
\]
To compute the variance of $X_v$, we write $X_v=\sum_{A\subseteq S:|A|=k-1} X_{A}$, where the random variable $X_A$ is $1$ when the vertices in $A$ induce a cycle of length $k$ with $v$ and $0$ otherwise.
\begin{align*}
\var{X_v}
&\leq \E(X_v)\\
&\ \ +\sum_{A,B\subseteq S:|A|=|B|=k-1,A\neq B} \cov{X_A,X_B}
\end{align*}

Now, for any fixed subsets $A,B\subseteq S$, $|A|=|B|=k-1$ and $|A\cap B|=r$, $\cov{X_A,X_B}\leq p^{2k-r}$ and the number of such subsets is at most $\binom{|S|}{2k-2-r}\binom{k}{r}\leq \binom{n}{2k-2-r}\binom{k}{r}$. Therefore,
\[
\sum_{r=0}^{k-2} \sum_{A,B\subseteq S:|A|=|B|=k-1,|A\cap B|=r} \frac{\cov{X_A,X_B}}{\E(X_v)^2}
\]
\begin{align*}
& \leq \sum_{r=0}^{k-2} \frac{\binom{k}{r}\binom{n}{2k-2-r}p^{2k-r}}{\binom{(1-\delta)n/10}{2k-2}p^{2k}}\\
& \leq \sum_{r=0}^{k-2}\frac{C_r}{(np)^{r}} \\
&\ \ \quad \quad (\text{for some constants $C_r$ dependent on $r,\delta$})\\
& \rightarrow 0
\end{align*}
as $n\rightarrow \infty$ if $p\geq C/n^{1-2/k}$ for some sufficiently large constant $C$ since each term in the summation tends to $0$ and the summation is over a finite number of terms. Thus
\[
\prob{X_v=0}\leq \frac{1}{\binom{(1-\delta)n/10}{k-1}p^k} \leq \frac{1}{((1-\delta)n/10)^{k-1}p^k} .
\]
Therefore,
\[
\prob{X_v\geq 1}\geq 1-\frac{1}{((1-\delta)n/10)^{k-1}p^k}
\]
and hence
\begin{align*}
\prob{X_v\geq 1 \forall v\in P} &\geq \left(1-\frac{1}{((1-\delta)n/10)^{k-1}p^k}\right)^{|P|}\\
&= \left(1-\frac{1}{((1-\delta)n/10)^{k-1}p^k}\right)^{\delta n}\\
&\geq e^{-\frac{10^{k-1}\delta}{2(1-\delta)^{k-1}n^{k-2}p^k}} \rightarrow 1
\end{align*}
as $n\rightarrow \infty$ if $p\geq \frac{C}{n^{1-2/k}}$ for some large constant $C$.
\end{proof}

Finally, we prove Lemma \ref{lemma:expected-no-of-cycles} by computing the expectation.
\begin{proof}[Proof of Lemma \ref{lemma:expected-no-of-cycles}]
$\E(\text{Number of cycles of length $k$})$
\begin{align*}
&\leq \sum_{i=1}^k \binom{|P|}{i}\binom{|R|}{k-i}k!p^k\\
&=\sum_{i=1}^k \binom{\delta n}{i}\binom{(1-\delta)n}{k-i}k!p^k\\
&\leq \sum_{i=1}^k (\delta n)^i((1-\delta)n)^{k-i}(kp)^k
\end{align*}
\begin{align*}
&= ((1-\delta)nkp)^k\sum_{i=1}^k \left(\frac{\delta}{1-\delta}\right)^i\\
&= ((1-\delta)nkp)^k(1-\delta) \leq (nkp)^k.
\end{align*}
\end{proof}

\section{Conclusion}
Several well-known combinatorial problems can be reformulated as hitting set problems with an exponential number of subsets to hit. However, there exist efficient procedures to verify whether a candidate set is a hitting set and if not, output a subset that is not hit. We introduced the implicit hitting set as a framework to encompass such problems. The motivation behind introducing this framework is in obtaining efficient algorithms where efficiency is determined by the running time as a function of the size of the ground set. We initiated the study towards developing such algorithms by showing an algorithm for a combinatorial problem that falls in this framework -- the feedback vertex set problem on random graphs. It would be interesting to extend our results to other implicit hitting set problems mentioned in Section 1.1.

\bibliographystyle{amsalpha}
\bibliography{references}
\end{document}